\documentclass{llncs}
\usepackage{amssymb,amsmath,graphicx,xspace,cite}

\newif\ifim \imtrue  

\let\doendproof\endproof
\renewcommand\endproof{~\hfill\qed\doendproof}

\newcommand{\e}[1]{\emph{#1}}
\renewcommand{\cal}[1]{\ensuremath{\mathcal{#1}}\xspace}
\renewcommand\-{\textrm{-}}

\renewcommand{\O}[1]{\ensuremath{O(#1 \log n)}\xspace}

\renewcommand\P{\ensuremath{P}\xspace}
\renewcommand{\d}[2]{\ensuremath{\mathsf{d}(#1,#2)}\xspace}
\newcommand{\dpr}[2]{\ensuremath{\mathsf{d'}(#1,#2)}\xspace}
\newcommand{\spm}[1]{\ensuremath{\textrm{SPM}(#1)}\xspace}

\renewcommand\ss{\ensuremath{s^*}\xspace}
\renewcommand\tt{\ensuremath{t^*}\xspace}

\newcommand{\uone}{\ensuremath{u_1}\xspace}\newcommand{\utwo}{\ensuremath{u_2}\xspace}\newcommand{\uthree}{\ensuremath{u_3}\xspace}
\newcommand{\vone}{\ensuremath{v_1}\xspace}\newcommand{\vtwo}{\ensuremath{v_2}\xspace}\newcommand{\vthree}{\ensuremath{v_3}\xspace}

\newcommand\pone{\ensuremath{\pi_1}\xspace}\newcommand\ptwo{\ensuremath{\pi_2}\xspace}\newcommand\pthree{\ensuremath{\pi_3}\xspace}

\newcommand{\A}{\ensuremath{\cal{A}}\xspace}

\newcommand\bae{\cite{bae}\xspace}

\title{Geodesic Diameter of a Polygonal Domain\\ in \O{n^{4}} Time}
\author{Mikko Koivisto \and Valentin Polishchuk}
\institute{Helsinki Institute for Information Technology HIIT\\ 
Department of Computer Science, University of Helsinki, Finland\\
\texttt{\{firstname.lastname\}@cs.helsinki.fi}}
\begin{document}\maketitle
\begin{abstract} We show that the geodesic diameter of a polygonal domain with $n$ vertices can be computed in \O{n^{4}} time by considering $O(n^{3})$ candidate diameter endpoints; the endpoints are a subset of vertices of the overlay of shortest path maps from vertices of the domain.\end{abstract}
\section{Introduction}
For many geometric shortest path problems efficient solutions have been developed both for simple polygons and for polygonal domains with holes, as well as for surfaces of polytopes in 3D \cite{survey}. A notable exception is computing the diameter. The problem is non-trivial already in simple polygons, where it was examined decades ago \cite{chazelleCutting,suriSimple} culminating in a linear-time algorithm \cite{hsMatrix}. Similarly, for convex polytopes polynomial-time algorithms have been known since 1990's \cite{orourke,star}; the current best running time is \O{n^7} (for a polytope with $n$ vertices) \cite{atlas}. However, for polygonal domains with holes, no algorithms existed until very recently. (The survey \cite{survey} mentions "brute-force" results, but no details are given.)

In \bae, Bae et al.\ presented an \O{n^{7.73+\varepsilon}}-time algorithm for computing the diameter of an $n$-vertex polygonal domain \P.\footnote{Earlier, in \cite{japan}, the same authors announced an $O(n^{5+\frac{10}{11}+\varepsilon})$-time algorithm.} Each of the diameter endpoints may be a vertex of \P, a point in the relative interior of \P's edge, or an interior point of \P. If either of the endpoints is a vertex, the diameter can be found in \O{n^2} by computing the shortest path maps \cite{survey,hs} from all vertices. However, Bae et al.\ \bae exhibit examples where the diameter is realized by non-vertex points on the boundary of \P or by a pair of interior points. According to \bae, "This observation also shows an immediate difficulty in devising any exhaustive algorithm since the search space like $\partial\P$ or the whole domain \P is not discrete."

In this paper, we show that there actually does exist a discrete set of candidate diameter endpoints: an $O(n^{3})$-size subset of vertices of the overlay of shortest path maps from vertices of the domain. For each candidate, the farthest point of \P can be found in \O{n} time by building the shortest path map.

\section{Preliminaries}

Let \P be a polygonal domain with $n$ vertices. The \e{shortest path map} from a
source point $p\in P$ \cite{survey} is the decomposition of \P into cells with
the same unique
combinatorial type of the shortest paths from $p$ to the points in a cell (the
combinatorial type of a path is the sequence of vertices visited by the path).
We use SPM as a shorthand for "shortest path map"; the SPM from a specific
source $p$ is denoted by \spm{p}. A \e{bisector} in SPM is the boundary between
two cells, i.e., an edge of the SPM. The vertices of the SPM are of three types:
vertices of \P, intersections of bisectors with sides of \P, and \e{triple
points} where 3 or more bisectors meet. The complexity of the SPM is $O(n)$, and
it can be built in \O{n} time \cite{hs}.

Let $\cal{A}$ be the overlay of the SPMs from all vertices of \P. The overlay
can be built in $O(n^4)$ time since there are $O(n)$ edges in the SPM from each
of the $n$ vertices, and two edges intersect $O(1)$ times. We will use the term
\e{node} for a vertex of $\cal{A}$, to distinguish nodes from vertices of \P.

Let \d{p}{q} denote the geodesic distance within \P between points $p,q$, and
let $\ss,\tt$ be the endpoints of the diameter of \P: $\d{\ss}{\tt}=\max_{p,q\in
P}\d{p}{q}$. As in \bae, we separately consider the cases when both of $\ss,\tt$
are interior, when one of $\ss,\tt$ is on the boundary, and when both are on the
boundary. The hardest, bottleneck case is when both diameter endpoints are
interior. In this case, the algorithm of \bae makes use of the following:
\begin{lemma}\label{lem5} \e{\bae}
If both \ss and \tt are in the interior of \P, there are at least 5
homotopically different shortest \ss-\tt paths. In addition, there exist at
least 3 distinct vertices $u_1,u_2,u_3$ that are adjacent to \ss on the shortest
\ss-\tt paths; similarly, there are at least 3 distinct vertices $v_1,v_2,v_3$
adjacent to \tt on the shortest paths: $\d{\ss}{\tt}=|\ss u_i|+\d{u_i}{\tt}=|\tt
v_j|+\d{v_j}{\ss}$ for $i,j=1,2,3$.
\end{lemma}
The above properties are analogous to the properties of the diameters of
polyhedral surfaces in $\mathbb{R}^3$, established by O'Rourke and Schevon
\cite{orourke}. For a formal proof of the properties refer to \bae. An intuitive
explanation of the 5-path property is that each of $\ss,\tt$ has 2 coordinates
(degrees of freedom), and the equation that two shortest paths are the same
length takes away one degree of freedom; thus, 5 path length equalities (i.e., 4
equations) take away all 4 degrees of freedom, pinning \ss and \tt. The
3-vertices property is even more intuitive: were there only 2 vertices adjacent
to \ss through which the shortest \ss-\tt paths go, one could move \ss away from
both of them, increasing the diameter.

Based on Lemma~\ref{lem5}, Bae et al.\ \bae scroll through all 5-tuples of vertices of \P, and for each 5-tuple look at each of the $O(n^2)$ cells in the overlay of the SPMs from the 5 vertices; within one cell, a constant-size system of constant-degree equations is solved to obtain candidate diameter endpoints. For each of the $O(n^7)$ candidate pairs of points, a two-point shortest path query is performed in \O{n^{8/11}} time \cite{2point}; the data structure for the queries can be built in $O(n^{7+8/11+\varepsilon})$ time which is the ultimate running time of the algorithm of \bae.

Our algorithm identifies an $O(n^{3})$-size set of candidate diameter endpoints. For each candidate, we simply find the farthest point in \P with the \O{n}-time algorithm of \cite{hs}.

\subsection{Overview of our approach}

We say that a vertex $u$ is the \e{first bend} of an \ss-\tt path if the path
starts from the segment $\ss u$; similarly $v$ is the \e{last} bend if the path
ends with $v\tt$. (Recall that none of $\ss,\tt$ is a vertex of \P.) We start from
proving a little variation of Lemma~\ref{lem5}: If both $\ss,\tt$ are interior,
then there exist \e{exactly} 3 vertices $\{u_1,u_2,u_3\}$, \e{exactly} 3
vertices $\{v_1,v_2,v_3\}$, and \e{at least} 5 homotopically different shortest
\ss-\tt paths such that $\{u_1,u_2,u_3\}$ are the first bends of the 5 paths,
and $\{v_1,v_2,v_3\}$ are the last bends. (We only need the existence of the
vertices and the paths; algorithmically we do not scroll through all triples of
vertices as potential bend points.)

We then define a graph $G$ on $\ss,\tt,u_1,u_2,u_3,v_1,v_2,v_3$ that reflects
the way \ss and \tt are connected with the 5 paths. A simple case analysis of
the connectivity of $U=\{u_1,u_2,u_3\}$ and $V=\{v_1,v_2,v_3\}$ in $G$ shows
that each of $\ss,\tt$ is either a triple point in the SPM from a vertex of \P,
or is at the intersection of bisectors in SPMs from two vertices (i.e., at a
node of $\cal{A}$).

Moreover, in the latter case, the vertices of $G$ are in a special relation; we prove that there are only $O(n^{3})$ sets of vertices that can be in the relation. Thus there are only $O(n^{3})$ candidate diameter endpoints; for each, farthest point in \P can be found in \O{n} time.

\subsection{Triple-point-diameter-end -- an easy case}

Suppose that there exists a diameter whose (at least one) endpoint is a triple
point in the SPM from one of the vertices of \P; we say that such point is a
\e{triple-point-diameter-end}. In this case the diameter can be computed in
\O{n^3} time since the total number of triple points in all SPMs in $O(n^2)$,
and for each triple point the farthest point of \P can be found in \O{n} time.
In what follows we will assume that there is no diameter, one endpoint of which
is a triple-point-diameter-end, since if it is the case we can compute it in
\O{n^3 \log n} time.

\subsection{An imaginary perturbation}

Lemma~\ref{lem5} states that the set of the first bends of the (at least 5) shortest paths defining the diameter, has \e{at least} 3 vertices. The statement does not exclude the possibility that there exist, say, 5 distinct vertices each being the first bend of one of the 5 (or 6, or 7, or more) different shortest \ss-\tt paths. (In fact, it is easy to concoct instances with an arbitrarily large set of possible first bends of the diameters.) The next lemma shows that there actually exist \e{exactly} 3 vertices that serve as first bends for at least 5 shortest \ss-\tt paths.
\begin{lemma}\label{pert}There exist 3 vertices $u_1,u_2,u_3$ and 5 homotopically different shortest \ss-\tt paths such that the set of the first bends of the 5 paths is $\{u_1,u_2,u_3\}$.\end{lemma}
\begin{proof} 
By Lemma~\ref{lem5}, the set $S$ of the first bends of the shortest \ss-\tt paths contains at least 3 vertices; clearly, \ss is in the convex hull of $S$ (or else \ss can be moved increasing the diameter). Suppose that there are actually more than 3 vertices in $S$. By Carath\'eodory's theorem, all but 3 points in $S$ are "redundant" in the sense that there exist 3 points $u_1,u_2,u_3\in S$ such that \ss is in the convex hull of $u_1,u_2,u_3$.

Let $u$ be a vertex of \P in $S\setminus\{u_1,u_2,u_3\}$. Let $\P'$ be \P with $u$ replaced by a little notch so as to increase the length of paths that bend at $u$ (Fig.~\ref{notch}); let \dpr{p}{q} be the geodesic distance within $\P'$. Because, clearly, $u$ is not on a shortest path from \tt to any of $u_1,u_2,u_3$, we have that $\dpr{\ss}{\tt}=\d{\ss}{\tt}$.
\begin{claim} \ss-\tt is a diameter of~$\P'$.\end{claim}
\begin{proof} Suppose that there exists a pair of points $a,b\in\P'$ such that $\dpr{a}{b}>\dpr{\ss}{\tt}$. Then \e{all} $a\-b$ paths must go through $u$. (Otherwise, if an $a\-b$ path $\pi$ does not go through $u$ then the lengths of $\pi$ in \P and $\P'$ are the same, and thus $\d{a}{b}=\dpr{a}{b}>\dpr{\ss}{\tt}=\d{\ss}{\tt}$ contradicting the fact that \ss-\tt is a diameter of \P.) If $a\-b$ were not a diameter in \P (i.e., if $\d{a}{b}<\d{\ss}{\tt}$), then we could perturb $u$ by a small enough amount to make sure $a\-b$ is also not the diameter in \P'; thus, $a\-b$ must be a diameter in \P. But since there are at least 3 paths from $a$ to $b$, there are also at least 3 paths from $u$ to $a$, i.e., $a$ is a triple point in \spm{u}, or is a triple-point-diameter-end -- a contradiction.
\end{proof}
In $\P'$, none of shortest \ss-\tt paths goes via $u$; thus the set $S$ of the first bends of \ss-\tt diameters has decreased. We can continue this way, decreasing $S$ until it has only 3 vertices. By Lemma~\ref{lem5}, there still exist 5 \ss-\tt paths (in the perturbed \P). But none of the 5 paths uses any of the perturbed vertices; hence the paths are the same in the perturbed \P as they were in \P\ -- with $\{u_1,u_2,u_3\}$ as the set of the first bends.
\begin{figure}\centering
\ifim\includegraphics[page=1]{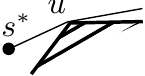}\quad$\rightarrow$\qquad\includegraphics[page=2]{notch}\fi
\caption{Perturb \P to remove shortest \ss-\tt paths not going through one of $u_1,u_2,u_3$.}
\label{notch}\end{figure}
\end{proof}
We emphasize that the perturbation in the above proof is {\bf imaginary}, not {\bf algorithmic} (symbolic, random, or otherwise). Our algorithm does not use the existence of exactly 3 or more than 3 first bend vertices. We need Lemma~\ref{pert} only to prove the correctness of our solution; the algorithm itself does not have to find the 3 vertices, nor to perturb \P to $\P'$, nor to do anything else according to the lemma.

\subsection{5 "independent" paths are necessary}

The 5-paths property ensures that whenever $(\ss,\tt)$ is a "diametrically
maximal" pair \bae (i.e., local motion of $\ss,\tt$ cannot increase the geodesic
distance between them), there exist 5 shortest \ss-\tt paths. However, the
converse does not hold automatically. That is, just mere existence of 5 shortest
\ss-\tt paths does not make $(\ss,\tt)$ diametrically maximal; it is also
important that no subset of the 5 paths could be obtained as a concatenation of
a smaller number of (sub)paths.

In particular, suppose that there exist 3 shortest \ss-\tt paths
\pone,\ptwo,\pthree having different first and last bends
\uone,\utwo,\uthree,\vone,\vtwo,\vthree. Moreover, suppose that \pone,\ptwo
intersect other than at $\ss,\tt$. Of course, two shortest \ss-\tt paths cannot
properly cross. By "intersection" we mean that \pone,\ptwo partially overlap,
sharing part of the way; i.e. that the paths have at least one common vertex $v$
(Fig.~\ref{share}).
\begin{figure}\centering
\ifim\includegraphics{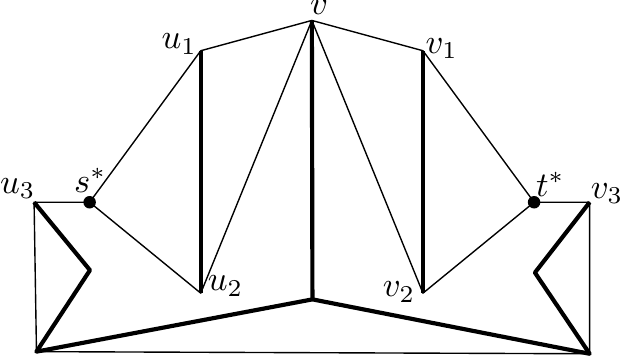}\fi
\caption{Thick segments are the obstacles. The 5 shortest paths are  $\pone=\ss\textrm{-}\uone\textrm{-}v\textrm{-}\vone\textrm{-}\tt,\ptwo=\ss\textrm{-}\utwo\textrm{-}v\textrm{-}\vtwo\textrm{-}\tt,\pthree=\ss\textrm{-}\uthree\textrm{-}\vthree\textrm{-}\tt,\pi_4=\ss\textrm{-}\uone\textrm{-}v\textrm{-}\vtwo\textrm{-}\tt,\pi_5=\ss\textrm{-}\utwo\textrm{-}v\textrm{-}\vone\textrm{-}\tt$. The figure is not to scale. 
}\label{share}
\end{figure}

In this case, there exist two more shortest \ss-\tt paths:
$\pi_4=\ss\-\uone\-v\-\vtwo\-\tt$ and $\pi_5=\ss\-\utwo\-v\-\vone\-\tt$.
Nevertheless, even though there exist 5 shortest \ss-\tt paths, \ss-\tt is
\e{not} a diameter: the 4 paths length equality,
$|\pone|=|\ptwo|=|\pi_4|=|\pi_5|$, are essentially only two equalities:
$|\ss\uone| + \d{\uone}{v} = |\ss\utwo| + \d{\utwo}{v}, |\tt\vone| +
\d{\vone}{v} = |\tt\vtwo| + \d{\vtwo}{v}$. These equalities "eat up" only one
degree of freedom from each of $\ss,\tt$ (putting \ss on the \uone-\utwo bisector
and \tt on the \vone-\vtwo bisector in \spm{v}). Equating the common length of
$\pone,\ptwo,\pi_4,\pi_5$ to the length of \pthree, takes away another degree of
freedom from the pair $(\ss,\tt)$. Still, the pair retains one degree of
freedom, and hence, $\ss,\tt$ can be simultaneously locally moved so as to
increase the diameter.

For a formalization of the above degree-of-freedom argument, one may look at the
proof of Thm.~2 (Case (II)) in \bae. The end of p.~6 in \bae considers the case
when the number of functions whose equality define $(\ss,\tt)$ is less than 5;
or equivalently, when the number of equations for \ss and \tt is less than 4.
This is exactly our case, as we have 3 equalities for the paths lengths. Bae et
al.\ prove that in this case the pair $(\ss,\tt)$ cannot give a local maximum of
the geodesic distance.

We summarize the above discussion in the following lemma:

\begin{lemma}\label{indep}Suppose that among the 5 shortest \ss-\tt paths there
exist two such that the first bends of the paths are different, the last bends
of the paths are different, and the paths intersect (overlap) other than at
$\ss,\tt$. Then \ss-\tt is not a diameter.\end{lemma}

\section{The diametric schema and an \O{n^5} algorithm}

We now give algorithms to compute the diameter in \O{n^5} time based on
Lemmas~\ref{pert} and \ref{indep}. In the
next section, we will carry out a more careful analysis of the bottleneck case
of this algorithm to reduce the runtime to \O{n^4}.

\subsection{Both \ss and \tt are interior}

We first consider the case when both $\ss,\tt$ are interior points of \P (the bottleneck case). By Lemma~\ref{pert}, there exist 5 shortest \ss-\tt paths, and triples of vertices $U=\{u_1,u_2,u_3\}$ and $V=\{v_1,v_2,v_3\}$ such that all 5 paths have one of $U$ as the first bend and one of $V$ as the last bend. Consider the graph $G$ on $\ss,\tt,U,V$, with edges between \ss and each of $U$, between \tt and $V$, and between $u\in U$ and $v\in V$ whenever one of the 5 paths goes via $u$ and via $v$ (Fig.~\ref{schema}). We call $G$ the \e{diametric schema} because it shows how \ss and \tt are interconnected by the diameters. We define the \e{degree} of $u\in U$ to be the number of its neighbors in $V$; the degree of $v\in V$ is the number of $v$'s neighbors in $U$.
\begin{figure}\centering
\ifim (a)\includegraphics[page=1]{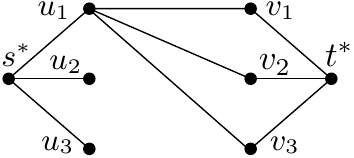}\hfil(b)\includegraphics[page=2]{schema}\hfil(c)\includegraphics[page=3]{schema}\fi
\caption{(a) If $u\in U$ has degree 3, \tt is a triple point of \spm{u}. (b) If $u_1$ and $u_2$ have the same 2 neighbors, $G$ is not planar. (c) The neighbors of $u_2$ are $v_2,v_3$; w.l.o.g., $u_3$ is connected to $v_3$.}
\label{schema}\end{figure}

We will now go through all possible interconnection patterns between $U$ and $V$. First, suppose that a vertex $u\in U$ has degree 3 (Fig.~\ref{schema}a). In this case, there are 3 homotopically different shortest paths from $u$ to \tt: via $v_1$, via $v_2$ and via $v_3$. That is, \tt is a triple-point-diameter-end.

Assume now that no vertex in $U,V$ has degree 3 in the schema. Since the total degree of the vertices in $U$ is 5, two of the vertices, say, $u_1,u_2$, are degree-2. Let $v_1,v_2$ be the neighbors of $u_1$.

First consider the case when $v_3$ is \e{not} a neighbor of $u_2$. That is, the
neighbors of $u_2$ are $v_1,v_2$, and hence, the neighbor of $u_3$ is $v_3$
(Fig.~\ref{schema}b). Then $G$ is not planar: by contracting the edges $\ss
u_3,\tt v_3$, we obtain $K_{3,3}$ (with $\ss,v_1,v_2$ in one part, and
$\tt,u_1,u_2$ in the other) as the schema's minor. However, the edges
$\ss\uone,\ss\utwo,\ss\uthree,\tt\vone,\tt\vtwo,\tt\vthree$ are
pairwise-non-crossing. Hence, there exist two shortest \ss-\tt paths such that
the first bends of the paths are different, the last bends of the paths are
different, and the paths intersect (overlap) other than at $\ss,\tt$
(Fig.~\ref{share}). By Lemma~\ref{indep}, \ss-\tt is not a diameter in this
case.


In the remaining case, $v_3$ is a neighbor of $u_2$. That is, only one of $v_1,v_2$ (say, $v_2$) is a neighbor of $u_2$; the other neighbor of $u_2$ is $v_3$ (Fig.~\ref{schema}c). This leaves two possibilities of connecting $u_3$ to $V$: either to $v_1$ or to $v_3$. Both possibilities result in the same (up to isomorphism) subgraph of the schema on $U,V$ -- a 5-edge path through $U\cup V$. W.l.o.g.\ we will assume that this path is $v_1\-u_1\-v_2\-u_2\-v_3\-u_3$, as in Fig.~\ref{schema}c.

We are almost done: it follows from the schema that \ss is on a bisector in
\spm{v_2} \e{and} is on a bisector in \spm{v_3}; i.e., \ss is a node of the
overlay $\cal{A}$. Thus, we can go through all nodes of $\cal{A}$, and find the
furthest point of \P for each. Since the overlay contains $O(n^4)$ nodes, we find
the diameter in \O{n^5} time.

\subsection{Both \ss and \tt are boundary}
When both $\ss,\tt$ are on the boundary of \P, each of them has one degree of freedom: \cite[Theorem~2]{bae} proves that there must be \e{at least} 2 vertices that serve as first bends on the shortest \ss-\tt paths, and 2 vertices that serve as last bends; moreover, to pin both \ss and \tt, there must exist at least 3 homotopically different shortest \ss-\tt paths. Similarly to Lemma~\ref{pert}, we can assume that there exist \e{exactly} two vertices $u_1,u_2$, \e{exactly} two vertices $v_1,v_2$, and 3 shortest \ss-\tt paths such that $U=\{u_1,u_2\}$ is the set of first bends of the paths and $V=\{v_1,v_2\}$ is the set of the last bends. Hence, w.l.o.g.\ the diametric schema looks as in Fig.~\ref{bd}a. This means that \ss is a vertex in \spm{v_2} as it lies on the intersection of a bisector in the map and an edge of \P. Because there are $O(n^2)$ such vertices, the diameter can be found in \O{n^3} time.
\begin{figure}\centering
\ifim (a)\includegraphics{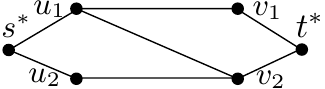}\hfil(b)\includegraphics{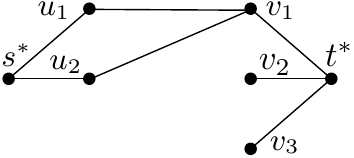} \fi
\caption{(a) The schema for the case of both $\ss,\tt$ being boundary: \ss is a vertex of \spm{v_2}. (b) \ss is boundary, \tt interior: \ss is a vertex of \spm{v_1}.}
\label{bd}\end{figure}
\subsection{\ss is boundary, \tt is interior}
Finally, if \ss is on the boundary of \P and \tt is in the interior, there exist two first-bend vertices $u_1,u_2$, three last-bend vertices $v_1,v_2,v_3$, and at least 4 shortest \ss-\tt paths -- this is proved in \cite[Theorem~2]{bae} and can also be seen by the degrees-of-freedom argument. Thus, one of $v_1,v_2,v_3$ (say, $v_1$) has degree 2 in the schema (Fig.~\ref{bd}b), and \ss is a vertex in \spm{v_1} (foot of a bisector touching an edge of \P). Because overall there are $O(n^2)$ vertices in the SPMs from vertices of \P, the diameter can be found in \O{n^3} time.


\section{Plausible vertices and an \O{n^{4}} algorithm}\label{faster} 

The bottleneck case in the algorithm given in the previous section 
is when both $\ss,\tt$ are interior and the diametric schema
is as in Fig.~\ref{schema}c. The running time turned \O{n^5} because we scrolled
through \e{all} $O(n^4)$ nodes of $\cal{A}$. However, according to the schema,
\ss cannot be at an arbitrary node: it is defined by the intersection of a
bisector between $u_1,u_2$ in \spm{v_2} and a bisector between $u_2,u_3$ in
\spm{v_3}, while $v_2,v_3$ define a bisector in \spm{u_2}. We now show that
there are only $O(n^{3})$ nodes with such properties.

Say that vertices $i,j$ are \e{neighbors} in \spm{k} if there is a bisector between $i,j$ in \spm{k}.
\begin{definition} A node $p$ of $\cal{A}$ is \e{plausible} if there exist vertices $u_1,u_2,u_3,v_2,v_3$ of \P such that
\begin{itemize}
\item $p$ is at the intersection of the bisector between $u_1,u_2$ in \spm{v_2} and the bisector between $u_2,u_3$ in \spm{v_3};
\item $v_2,v_3$ are neighbors in \spm\utwo. 
\end{itemize}\end{definition}
\begin{definition} A 5-tuple $(u_1,u_2,u_3,v_2,v_3)$ of vertices of \P is \e{plausible} if
\begin{itemize}
\item $u_1,u_2$ are neighbors in \spm\vtwo; 
\item $u_2,u_3$ are neighbors in \spm\vthree; 
\item $v_2,v_3$ are neighbors in \spm\utwo. 
\end{itemize}\end{definition}
The number of plausible nodes is not larger than the number of plausible 5-tuples (because even if a 5-tuple is plausible, the bisectors between $u_1,u_2$ and $u_2,u_3$ may not intersect at all).
\begin{lemma}There are $O(n^{3})$ plausible 5-tuples.\end{lemma}
\begin{proof} Define an $n\times n\times n$ "bisectors adjacency" array $B=\{b_{ijk}\}$ as follows:
\[
b_{ijk} =
\begin{cases}
 1 & \textrm{if there is a bisector between vertices $i$ and $j$ in \spm{k}},\\
 0 & \textrm{otherwise}\,.
\end{cases}
\]
Here and throughout the indices run from 1 to $n$. 

A 5-tuple $(u_1,u_2,u_3,v_2,v_3)$ is plausible iff $b_{u_1u_2v_2}=b_{u_3u_2v_3}=b_{v_3v_2u_2}=1$. 
The number of plausible 5-tuples is
\[
\sum_{u_1u_2u_3v_2v_3}b_{u_1u_2v_2}b_{u_3u_2v_3}b_{v_3v_2u_2}
=\sum_{ijklm}b_{ijl}b_{kjm}b_{mlj}=\sum_j\sum_{ml}b_{mlj}\sum_{ik}b_{ijl}b_{kjm}\,.
\]
Since each SPM has linear complexity,
\[\sum_{ml}b_{mlj} = O(n)  \quad \textrm{for all $j$}\,.\]
That is, for any $j$, there are $O(n)$ pairs $(m_t^j,l_t^j)$ such that $b_{m_t^jl_t^j}=1$. Hence, the number of plausible 5-tuples is
\begin{multline}\notag
\sum_j\sum_{ml}b_{mlj}\sum_{ik}b_{ijl}b_{kjm} 
= \sum_j\sum_{t=1}^{O(n)}\sum_{ik}b_{ijl_t^j}b_{kjm_t^j} 
= \sum_{t=1}^{O(n)}\sum_{ijk}b_{ijl_t^j}b_{kjm_t^j} \le \\
\le \sum_{t=1}^{O(n)}\Big(\sum_{ij} b_{ijl_t^j}\Big)\Big(\sum_{kj} b_{kjm_t^j}\Big) 
= O(n)\cdot O(n) \cdot O(n) = O(n^3)\,.
\end{multline}
\end{proof}
Returning to our algorithm, for every node of $\cal{A}$, we can test in constant time whether it is plausible by checking the corresponding entries in $B$ (clearly, $B$ itself can be filled as the SPMs from the vertices of \P are built). For each of the $O(n^{3})$ plausible nodes, we find the farthest point in \O{n} time, and hence we have:
\begin{theorem}The diameter of \P can be found in \O{n^{4}} time.\end{theorem}

\section{Conclusion}

We showed how to compute the diameter of a polygonal domain in \O{n^4} time. A faster algorithm for the problem would have to use new insights: in our algorithm, already computing the arrangement \A takes $O(n^4)$ time.

An interesting open problem is whether our ideas can be applied to diameters of polytopes in $\mathbb{R}^3$ \cite{star,atlas,orourke}. Shortest paths on polyhedral surfaces do not bend at vertices \cite{ss}; the combinatorial type of a shortest path is the sequence of edges that it visits (the path is uniquely defined by the sequence due to the unfolding property -- the shortest path becomes a line segment if the polytope is unfolded along the edges in the sequence). The 5-diameters property holds for polytopes as well \cite{orourke}. If \ss belongs to the interior of a face $f$ of the polytope, then the diameters bend on at least 3 edges bounding $f$. By a perturbation argument as in Lemma~\ref{pert}, there exists exactly three edges $u_1,u_2,u_3$ of $f$ at which the 5 paths bend; similarly, there exist 3 edges, $v_1,v_2,v_3$, of \tt's face at which the 5 diameters bend. Hence, one can define the diametric schema. However, here the analogy between polygonal domains and polytopes seems to end: were are not aware of a notion of a SPM from a polytope edge.

The existence of several homotopically different paths between the diameter endpoints suggests to study properties of the 2nd, 3rd, and in general, $K$th homotopically different shortest paths between two points in a polygonal domain. Even though algorithmically we do not use the paths to compute the diameter, it seems interesting to study their combinatorial properties. How can the "$K$th SPM" be represented and what is its complexity?
\paragraph{Acknowledgements} We thank Sergey Bereg, Atlas Cook IV, Irina Kostitsyna, Joe Mitchell, Bettina Speckmann and Kevin Verbeek for discussions and comments. 
\bibliographystyle{abbrv}\bibliography{ksp}
\end{document}